\documentclass[11pt]{article}

\usepackage{amsmath,fullpage}
\usepackage{xspace}
\usepackage{amssymb}
\usepackage{epsfig}

\newtheorem{Thm}{Theorem}

\newtheorem{Cor}{Corollary}

\newtheorem{Claim}{Claim}

\newtheorem{Fact}{Fact}
\newenvironment{proof}{\noindent {\textbf{Proof }}}{$\Box$
\medskip}

\newcommand\E{{\mathbb{E}}}

\newcommand\mcP{\mathcal{P}}

\mathchardef\mhyphen="2D

\newcommand{\suppress}[1]{}
\newcommand\COMMENT[1]{}

\begin{document}
\title{Optimal Direct Sum Results for Deterministic and Randomized Decision Tree Complexity
\thanks{
Research supported in part by 
the French ANR QRAC project under contract ANR-08-EMER-012.
Research at the Centre
for Quantum Technologies is funded by the Singapore Ministry of Education 
and the National Research Foundation.}}
\author{Rahul Jain\thanks{Centre for Quantum Technologies and Department of Computer Science, National University of Singapore. Email: {\tt rahul@comp.nus.edu.sg}} \quad Hartmut Klauck\thanks{Centre for Quantum Technologies (NUS) and School of Physical and Mathematical Sciences, Nanyang Technological University. Email: {\tt hklauck@gmail.com}}  \quad Miklos Santha\thanks{CNRS - LRI, Universit\'e Paris-Sud, Orsay, France and Centre for Quantum Technologies, National University of Singapore. Email: {\tt santha@lri.fr}} }
\date{}
\maketitle

\begin{abstract}
A {\em Direct Sum Theorem} holds in a model of computation, when solving some $k$ input instances together is $k$ times as expensive as solving one. We show that Direct Sum Theorems hold in the models of deterministic and randomized decision trees for all relations. We also note that a near optimal Direct Sum Theorem holds for quantum decision trees for boolean functions.
\end{abstract}


\section{Introduction}

One of the goals of complexity theory is to understand the structural properties of different models of computation. A fundamental question that can be asked in every model of computation is how well different computations may be combined. Can we achieve substantial savings when solving the same problem on $k$ (independent) inputs together? Or is the straightforward approach, namely running the same algorithm $k$ times, optimal? This question is known as the {\em direct sum problem}, and has been studied in many different settings and variations.

We say that a {\em Direct Sum Theorem} holds for a measure of complexity, when solving $k$ input instances together is roughly as costly as $k$ times solving one instance according to that measure.
Since we are often interested in bounded error computations, we also need to specify how the error on $k$ input instances relates to the error on one instance. The direct sum question in a narrower sense relates to solving $k$ instances with constant error, while a {\em Strong Direct Product Theorem} holds when even using roughly $k$ times the resources required to solve one instance with constant error, the success probability goes down exponentially in $k$. This happens when we solve the $k$ instances independently, and a Strong Direct Product Theorem states that this is optimal with respect to resources and error. In this paper we only consider the direct sum problem in the narrower sense: we compare solving one instance with constant (resp.~no) error to solving $k$ instances with constant (resp.~no) error.

The decision tree model (see \cite{buhrman:dectreesurvey}) is perhaps the simplest model of computation, measuring the number of input positions that need to be accessed in order to compute a function/solve a relation. Still many questions about this model remain open. In this paper we show that the direct sum property holds for deterministic and randomized decision trees.

Previously, a Strong Direct Product Theorem for decision trees was established by Nisan et al.~\cite{nisan:dpt}. However, their result does not imply a Direct Sum Theorem in our sense, because it is only shown in a weaker setting. Instead of analyzing a single algorithm that has access to all $k$ inputs and produces all $k$ outputs, Nisan et al.~consider a setting where $k$ algorithms (that can access all inputs), each making at most $d$ queries, compute one of the $k$ outputs each, where $d$ is the query complexity of computing one instance (with bounded error). Hence this does not establish a Direct Sum Theorem in the above sense.

Previous papers \cite{nisan:dpt, ben-asher:dpt} have dismissed the direct sum problem for decision trees as either very simple, or uninteresting. To quote \cite{nisan:dpt}: ``While it is an easy exercise to see that ìdirect-sumî holds
for decision tree depth, the other two problems (direct product and help bits) are more difficult." The paper does not make it clear, what kind of decision tree is meant (the setting considered there is distributional complexity). In the distributional setting a general counterexample by Shaltiel \cite{shaltiel:sdpt} makes it clear that some very tight direct sum statements are not even true for the model where there is one decision tree that has to solve all $k$ input instances together.

Ben-Asher and Newman claim in \cite{ben-asher:dpt}: ``In the standard decision tree model the question is quite uninteresting as queries do not involve variables of more than one of the problem instances at a time."
This does not seem to be a valid assessment of the problem, because with the same argument the strong direct product question for decision trees could be dismissed, which is as of now still an open problem (in the setting where one algorithm makes all outputs).

We give proofs of Direct Sum Theorems for the case of deterministic and randomized decision trees. In the deterministic case the main problem is to construct a more efficient decision tree for one instance from a given tree for two instances. In the randomized case the proof is along the lines of some proofs of Direct Sum Theorems in communication complexity e.g.~\cite{JainRS05}.

One may ask if a similar result is true for quantum decision trees, also known as quantum query algorithms. While we do not have a proof for this model, a weaker statement can be derived from recent results by Reichardt.
In \cite{reichardt:span} he shows that the general quantum adversary bound is tight within a logarithmic factor for the quantum decision tree complexity of every boolean function (his Theorem 1.4). He also shows that the general adversary bound has a direct sum behavior (see Theorem 7.2 in the long version of the paper. Note that one has to choose a good ``connection" function $f$ like XOR because the adversary bound works only for boolean functions). The direct sum for the general adversary bound has also been shown previously in Ambainis, Childs, Le Gall and Tani~\cite{AmbainisCGT09}. Hence we can conclude that in the quantum case, at least within a logarithmic factor and for boolean functions, a Direct Sum Theorem also holds.

\section{Preliminaries} \label{prelim}
A {\em deterministic decision tree} on $m$ variables is a rooted binary tree $T$ whose internal vertices
are labeled by the boolean variables $x_1, \ldots , x_m$, and whose leaves are labeled by the output  values from a set $\cal Y$.
For every vertex $v$ in $T$, we denote by $v_0$ (respectively $v_1$) the
left son (respectively the right son) of $v$,
and by $T(v)$ the subtree of $T$ rooted at $v$.
We set $T_b = T(v_b)$, for $b\in \{0,1\}$, where $v$ is the root of $T$.
The depth $d_T(v)$ of vertex $v$ in tree $T$, is defined recursively: it is 0 if $v$ is a leaf, otherwise
$d_T(v) = \max \{d_T(v_0), d_T(v_1)\} +1$. The depth $d(T)$ of $T$ is simply the depth of its root.
Every tree naturally computes a function $f_T$ on $m$ variables,
whose value at an assignment $x = (x_1, \ldots , x_m) \in \{0,1\}^m$ is defined recursively as follows: If the root of
$T$ is a leaf, then $f_T(x)$ is the value of its label. Otherwise, if $x_i$ is the label of the root and $x_i =b$, then
$f_T(x) = f_{T_b}(x)
$.

Clearly, several decision tree compute the same function $f$. The {\em deterministic decision tree complexity}
of $f$ is the depth of the minimal depth decision tree $T$ such that $f_T = f$, and we denote it by $D(f)$.

The above definitions naturally extend to trees
whose leaves are labeled by elements of ${\cal Y}^k$, for some positive integer $k$.
We call these trees $k${\em -output} deterministic
decision trees, they compute $k$-{\em output} functions whose range is
by definition ${\cal Y}^k$. We will use the notation $f = (f^{(1)}, \ldots , f^{(k)})$ for
$k$-output functions, where
$f^{(i)}$ is the function computing the $i$th output of $f$. In particular, we are interested here in
the case when $m = kn$ and the functions
do not share common input variables. More precisely, let
$f : \{0,1\}^{kn} \rightarrow {\cal Y}^k$  be a $k$-output function
whose input variables are $x_{1,1}, \ldots , x_{1,n},  \ldots , x_{k,1}, \ldots ,  x_{k,n}$. We set ${\bar x}_i = (x_{i,1}, \ldots , x_{i,n})$, and say
that $f$ is $k$-{\em independent} if the value of
$f^{(i)}$ depends only on ${\bar x}_i$.

One can also extend the definition of deterministic decision trees and $k$-independence to relations $f\subseteq \{0,1\}^{m} \times {\cal Y}$ instead of functions in a straightforward way (decision trees are required to find an output $y$ for each input $x\in\{0,1\}^m$ so that $(x,y)\in f$). 

In particular, for a relation $f\subseteq \{0,1\}^{m} \times {\cal Y}^k$, the relation $f^{(i)}\subseteq \{0,1\}^{m} \times {\cal Y}$ consists of all $(x,y)$, such that $(x,y_1,\ldots, y_{i-1},y,y_{i+1},\ldots, y_k)\in f$ for some $y_1\ldots,y_{i-1},y_{i+1},\ldots, y_k$.

Note that for inputs $x$ for which there is no $y$ with $(x,y)\in f$ no requirement on the output is made, and hence we can assume that all relations are total without loss of generality. Since for each input only one output can be produced, each deterministic tree automatically computes a function that is consistent with the relation in question.

For a relation $f \subseteq\{0,1\}^{n} \times {\cal Y}$, we define
the $k$th {\em tensor power} of $f$ as the
relation
$f^{\otimes k} \subseteq \{0,1\}^{kn} \times{\cal Y}^k$ by
$f^{\otimes k} = \{ ((\bar x_1,\ldots,\bar x_k),(y_1,\ldots,y_k)): \forall
i: (\bar x_i,y_i)\in f\}.$ Note that
$f^{\otimes k}$ is $k$-independent.

A {\em randomized decision tree} on $m$ variables is a convex combination of deterministic decision trees, such that for each input $x$ a correct output is computed with probability $1-\epsilon$ for a given error probability $\epsilon$. If not mentioned otherwise $\epsilon=1/3$. For $k$-output relations $f$ an output $(y_1,\ldots,y_k)$ is considered erroneous, if $({\bar x}_i ,y_i)\not\in f^{(i)}$ for some $i$, i.e., all $k$ outputs are required to be correct simultaneously.

$R_\epsilon(f)$ denotes the $\epsilon$-error randomized query complexity of $f$, which is the maximum number of queries made by the best randomized decision tree with error being at most $\epsilon$ on any input. Let $\mu$ be a distribution on $\{0,1\}^n$. Let $R^\mu_\epsilon(f)$ represent the $\epsilon$-error distributional query complexity of $f$, which is the maximum number of queries made by the best randomized decision tree with average error at most $\epsilon$ under $\mu$ (note that such a tree can be assumed to be deterministic w.l.o.g., but sometimes it is simpler to give a randomized tree). We have the following fact from \cite{yao:prob}.
\begin{Fact}[Yao's Principle]
$R_\epsilon(f) = \max_\mu R^\mu_\epsilon(f) $.
\end{Fact}

\section{Direct Sum for Deterministic Complexity} \label{deterministic}
Let $f \subseteq \{0,1\}^{kn} \times  {\cal Y}^k$ be a $k$-output relation.
Obviously $D(f) \leq \sum_{i=1}^k D_{f^{(i)}}$ since the values $f^{(i)}$ can be evaluated sequentially.
We prove that for $k$-independent relations this is in fact the least
expensive way to evaluate $f$, 
that is the inverse inequality also holds.
\begin{Thm}[Deterministic Direct Sum] \label{theorem:deterministic}
For every $k$-independent relation $f \subseteq \{0,1\}^{kn} \times  {\cal Y}^k$, we have $D(f) \geq \sum_{i=1}^k D(f^{(i)})$. 
\end{Thm}

\begin{proof}
Let $T$ be a $k$-output deterministic decision tree on variables
$\{x_{1,1}, \ldots , x_{k,n}\}$. For $i = 1, \ldots, k$, we refer to $\{x_{i,1}, \ldots , x_{i,n}\}$ as the $i$th group
of variables.
For every vertex $v$ of $T$, we define recursively $k$
single output decision trees $T_1(v), \ldots ,T_k(v)$,
where the vertices of $T_i(v)$ are labeled by the variables from
the $i$th group.
If $v$ is a leaf with label $(b_1, \ldots , b_k)$, then
$T_i(v)$ is a single node tree (a leaf),
with label $b_i$. Otherwise, let $v$ be an internal node and let's suppose
that its label is from the $j$th group of variables.  
The root of $T_j(v)$ is by definition
$v$ with the same label as in $T$,
its left subtree is $T_j(v_0)$ and its right subtree is $T_j(v_1)$. For all $i \neq j$, the tree
$T_i(v)$ is defined as the shallower (smaller depth) tree between $T_i(v_0)$ and $T_i(v_1)$.

\begin{Claim} \label{claim:depth}
For every vertex $v$ of $T$, we have  $\sum_{i=1}^k d(T_i(v))  \leq d_T(v)$.
\end{Claim}

\begin{proof}
The proof is by induction on the depth of $v$, and the statement is obviously true when $v$ is a leaf.
We suppose without loss of generality that the label of $v$ is from the $j$th group.
Let $b \in \{0,1\}$ such that
$d(T_j(v)) = d(T_j(v_b)) + 1$.
By definition, for all $i \neq j$, we have
$d(T_i(v)) =  \min \{ d(T_i(v_0)), d(T_i(v_1))\}$, and therefore $d(T_i(v)) \leq d(T_i(v_b))$. Thus
\begin{eqnarray*}
\sum_{i=1}^k d(T_i(v)) & \leq  &
d(T_j(v_b)) + 1 +  \sum_{i \neq j} d(T_i(v_b)) \\
& \leq & d_T(v_b) +1  \\
& \leq  & d_T(v),
\end{eqnarray*}
where the second inequality follows from the
inductive hypothesis, and the third one from the definition of the depth.
\end{proof}

We say that $T$ is {\em parsimonious} if no variable appears twice on the same root-leaf path.

\begin{Claim} \label{claim:simultaneous}
Let $T$ be parsimonious. Then for every vertex $v$ in $T$, for every $1 \leq i \leq k$, and 
for every  assignment ${\bar x}_i \in \{0,1\}^n$
for the variables in the $i$th group, there exists, for all $j \neq i$, an assignment
${\bar x}_j \in \{0,1\}^n$ for the variables in the $j$th group such that
$$f_{T_i(v)}({\bar x}_i) = f_{T(v)}^{(i)}({\bar x}_1 , \ldots , {\bar x}_i , \ldots , {\bar x}_k).$$
\end{Claim}

\begin{proof}
The proof is again by induction on the depth of $v$. Fix $1 \leq i \leq k$. If $v$ is a leaf, we can choose
for every ${\bar x}_i  \in \{0,1\}^n$ an arbitrary ${\bar x}_j \in \{0,1\}^n$, for $j \neq i$.
Otherwise, we distinguish two cases, according to the label of $v$.

Case 1: The label of $v$ is $x_{i,p}$ from the $i$th group of variables, for some $1 \leq p \leq n$.
Let ${\bar x}_i  \in \{0,1\}^n$ be an assignment for the variables in the $i$th group, and let  $x_{i,p} = b$.
By the inductive hypothesis there exists ${\bar x'}_j$, for $j \neq i$, such that
$$f_{T_i(v_b)}({\bar x}_i ) = f_{T(v_b)}^{(i)}({\bar x'}_1, \ldots  , {\bar x}_i , \ldots , {\bar x'}_k).$$
We set ${\bar x}_j  = {\bar x'}_j $, for $j \neq i$. Then we have
\begin{eqnarray*}
f_{T_i(v)}({\bar x}_i ) & = &
f_{T_i(v_b)}({\bar x}_i ) \\
& = &  f_{T(v_b)}^{(i)}({\bar x'}_1, \ldots  , {\bar x}_i , \ldots , {\bar x'}_k)\\
& = &  f_{T(v)}^{(i)}({\bar x}_1, \ldots  , {\bar x}_i , \ldots , {\bar x}_k).
\end{eqnarray*}
The first equality follows from the definition of $f_{T_i(v_b)}({\bar x}_i )$ since
$x_{i,p} = b$. The third equality also holds because by definition
$ f_{T(v)}({\bar x}_1, \ldots  , {\bar x}_i , \ldots , {\bar x}_k) =
f_{T(v_b)}({\bar x}_1, \ldots  , {\bar x}_i , \ldots , {\bar x}_k)$.

Case 2: The label of $v$ is $x_{j,p}$ from the $j$th set of variables for some $j \neq i$ and $1 \leq p \leq n$.
Let $b$ be such that
$T_i(v) =T_i(v_b)$.
Then again by the inductive hypothesis, for every ${\bar x}_i  \in \{0,1\}^n$, there exists
${\bar x'}_j$, for $j \neq i$, that satisfy
$$f_{T_i(v_b)}({\bar x}_i ) = f_{T(v_b)}^{(i)}({\bar x'}_1, \ldots  , {\bar x}_i , \ldots , {\bar x'}_k).$$
We define $x_{l,q}$ for $l \neq i$ and $q = 1, \ldots, n$ by
$$
x_{l,q} =\begin{cases}
b & \text{ if }  (l,q) = (j,p), \\
x'_{l,q} & \text{ otherwise.}
\end{cases}
$$
Then, similarly to Case 1, we have the following series of equalities:
\begin{eqnarray*}
f_{T_i(v)}({\bar x}_i ) & = &
f_{T_i(v_b)}({\bar x}_i ) \\
& = &  f_{T(v_b)}^{(i)}({\bar x'}_1, \ldots  , {\bar x}_i , \ldots , {\bar x'}_k)\\
& = &  f_{T(v)}^{(i)}({\bar x}_1, \ldots  , {\bar x}_i , \ldots , {\bar x}_k).
\end{eqnarray*}
The first equality is true because $T_i(v) =T_i(v_b)$.
The path followed on input $({\bar x}_1, \ldots  , {\bar x}_i , \ldots , {\bar x}_k)$ in $T(v)$ goes
from $v$ to $v_b$ since $x_{j,p} =b$, and then it is identical to the path followed
on input $({\bar x'}_1, \ldots  , {\bar x}_i , \ldots , {\bar x'}_k)$
in $T(v_b)$ because $T$ is parsimonious. Therefore
$f_{T(v_b)}({\bar x'}_1, \ldots  , {\bar x}_i , \ldots , {\bar x'}_k)
= f_{T(v)}({\bar x}_1, \ldots  , {\bar x}_i , \ldots , {\bar x}_k)$, and
the last equality also holds.
\end{proof}

We now prove Theorem~\ref{theorem:deterministic} by contradiction.
Let us suppose that $D(f) < \sum_{i=1}^n D(f^{(i)})$.
Let $T$ be a deterministic decision tree of depth $D(f)$ which computes a function $\tilde f$ that is consistent with the relation $f$.
Since $T$ is a minimal depth decision tree computing $\tilde f$,
we can suppose without loss of generality that $T$ is parsimonious.
Let $r$ be the root of $T$, then $d(r) = D(f)$.
For $i = 1, \ldots, k$, let $T_i = T_i(r)$.
By Claim~\ref{claim:simultaneous} and $k$-independence, $T_i$ computes an $\tilde f^{(i)}$, which is consistent with $f^{(i)}$, and
therefore $D(f^{(i)}) \leq d(T_i)$. Thus
$d(r) < \sum_{i=1}^k d(T_i)$,
contradicting Claim~\ref{claim:depth}.
\end{proof}

\begin{Cor} \label{corollary:deterministic}
For every relation $f\subseteq \{0,1\}^{n} \times {\cal Y}$ and for every integer $k$, we have
$D({f^{\otimes k} }) = k \cdot D(f)$.
\end{Cor}

\section{Direct Sum for Randomized Query Complexity}

\begin{Thm}[Randomized Direct Sum]
Let $f \subseteq \{0,1\}^{n} \times  {\cal Y}$  be a relation. Let $k$ be a positive integer and let $\delta > 0$ be a small constant. Then $R_\epsilon(f^{\otimes k}) \geq \delta^2 \cdot k \cdot  R_{\epsilon'}(f)$, where $\epsilon'= \frac{\epsilon}{1 - \delta} + \delta$.
\end{Thm}
\begin{proof}
Let $c = R_\epsilon(f^{\otimes k}) $. Let $\mcP$ be a randomized protocol for $f^{\otimes k}$ with $c$ queries and worst case error at most $\epsilon$. Let $\mu$ be a distribution on $\{0,1\}^n$. Let $\mu^{\otimes k} $ represent the distribution on $\{0,1\}^{kn}$ which consists of $k$ independent copies of $\mu$. Now let us consider the situation when we provide input to $\mcP$ distributed according to $\mu^{\otimes k}$. In such a situation we can fix the random coins of $\mcP$ in a suitable manner to get another protocol $\mcP_1$ such that $\E_{(x_1 \ldots x_k) \leftarrow \mu^{\otimes k}}[e(x_1 \ldots x_k)] \leq \epsilon$, where $e(x_1 \ldots x_k)$ represents the error made by $\mcP_1$, which is now a deterministic protocol, on input $(x_1 \ldots x_k)$ where each $x_i \in \{0,1\}^n$ represents the input for the $i$th instance of $f$. For notional convenience we use $x_i$ here instead of $\overline{x}_i$ as used in the previous section. Let $q(x_1 \ldots x_k)$ represent the number of queries made by $\mcP_1$ on input $(x_1 \ldots x_k)$. For each $1 \leq i \leq k$, 
let $q_i(x_1 \ldots x_k)$ represent the number of queries made by $\mcP_1$ on $x_i$ on input $(x_1 \ldots x_k)$. 
Since $q(x_1 \ldots x_k) = \sum_{i=1}^k q_i(x_1 \ldots x_k)$,  we have,
\begin{eqnarray*}
c  & \geq & \E_{(x_1 \ldots x_k) \leftarrow  \mu^{\otimes k}}[q(x_1 \ldots x_k)] \\
& = & \E_{(x_1 \ldots x_k) \leftarrow  \mu^{\otimes k}}[\sum_{i=1}^k q_i(x_1 \ldots x_k)] \\
& = & \sum_{i=1}^k \E_{(x_1 \ldots x_k) \leftarrow  \mu^{\otimes k}}[q_i(x_1 \ldots x_k)]
\end{eqnarray*}
Therefore there exists $1 \leq j \leq k$ such that $\E_{(x_1 \ldots x_k) \leftarrow \mu^{\otimes k}}[q_j(x_1 \ldots x_k)] \leq \frac{c}{k}$. Without loss of generality let $j=1$. Using this and the fact $\E_{(x_1 \ldots x_k) \leftarrow  \mu^{\otimes k}}[e(x_1 \ldots x_k)] \leq \epsilon$, we can argue by standard applications of Markov's inequality that there exist $x_2' \ldots x_k' \in \{0,1\}^{kn-n}$ such
that $\E_{x_1 \leftarrow \mu}[q_1(x_1 x_2' \ldots x_k')] \leq \frac{c}{\delta k}$ and $\E_{x_1 \leftarrow  \mu}[e(x_1 x_2' \ldots x_k')] \leq \frac{\epsilon}{1 - \delta}$. Therefore fixing $x_2' \ldots x_k'$ in $\mcP_1$ naturally gives rise to a protocol $\mcP_2$ for $f$ with expected number of queries under $\mu$ being at most $\frac{c}{\delta k}$ and expected error under $\mu$ being at most $\frac{\epsilon}{1- \delta}$. Now let us consider a protocol $\mcP_3$ which proceeds exactly as $\mcP_2$ but terminates whenever the number of queries exceeds $\frac{c}{\delta^2k}$. Again, using Markov's inequality, it can be argued that the expected error of $\mcP_3$ under $\mu$ is at most $\epsilon'= \frac{\epsilon}{1 - \delta} + \delta$ and of course the maximum queries made by $\mcP_3$ is at most $\frac{c}{\delta^2k}$. Hence by definition $R^\mu_{\epsilon'}(f) \leq \frac{c}{\delta^2k}$. Since this is true for every distribution $\mu$ on $\{0,1\}^n$, we get from Yao's Principle the desired result as follows.
$$R_{\epsilon'}(f) = \max_\mu R^\mu_{\epsilon'}(f) \leq \frac{c}{\delta^2k} = \frac{R_\epsilon(f^{\otimes k})}{\delta^2k} \enspace .$$
\end{proof}

\section{Open Problems} \label{open}

We proved direct sum theorems for deterministic and randomized query complexity. Note that 
it is also very easy to establish 
the direct sum property for
nondeterministic query complexity (also known as certificate complexity). However, several related open problems remain:
\begin{enumerate}
\item The direct sum theorem in the randomized case loses a factor of $\delta^2$ in the lower bound, as well as an additive $\delta$ in the error bound. While at least the factor in the lower bound is unavoidable in the setting of distributional complexity according to a result by Shaltiel \cite{shaltiel:sdpt}, this might not be necessary in the worst case complexity setting.
    \item In the quantum case no tight result is known, and the result following from Reichardt's work holds only
    for boolean functions. Can a tight result be established, even for all relations?
    \item Establishing general strong direct product theorems is open for both the quantum and the randomized/distributional setting. Note that the result of \cite{nisan:dpt} holds only in the weaker model where $k$ algorithms compute one output each.
\end{enumerate}

\bibliographystyle{alpha}

\begin{thebibliography}{BNRW03}

\bibitem[ACGT09]{AmbainisCGT09}
A.~Ambainis, A.~M.~Childs, F.~Le Gall, S.~Tani
\newblock The quantum query complexity of certification.
\newblock Preprint, arXiv:0903.1291.

\bibitem[BN95]{ben-asher:dpt}
Y. Ben-Asher and I. Newman.
\newblock Decision Trees with AND, OR Queries.
\newblock  In {\em Proceedings of 10th IEEE Conference Structure in Complexity Theory}, pp. 74--81, 1995.

\bibitem[BW02]{buhrman:dectreesurvey}
H.~Buhrman and R.~{de} Wolf.
\newblock Complexity measures and decision tree complexity: A survey.
\newblock {\em Theoretical Computer Science}, 288(1):21--43, 2002.

\bibitem[JRS05]{JainRS05} 
R.~Jain, J.~Radhakrishnan and P.~Sen.
\newblock Prior entanglement, message compression and privacy in quantum communication. 
\newblock In  {\em Proceedings of 20th IEEE Conference on Computational Complexity}, pp. 285--296, 2005.
\bibitem[NRS94]{nisan:dpt}
N.~Nisan, S.~Rudich and M.E.~Saks.
\newblock Products and Help Bits in Decision Trees.
\newblock In  {\em Proceedings of 35th IEEE Symposium on Foundations of Computer Science}, pp. 318--329, 1994.


\bibitem[R09]{reichardt:span}
B. W. Reichardt.
\newblock Span programs and quantum query complexity: The general adversary bound is nearly tight for every boolean function.
\newblock In {\em Proceedings of 50th IEEE Symposium on Foundations of Computer Science}, pp. 544--551, 2009. Long version under quant-ph/0904.2759

\bibitem[Sha01]{shaltiel:sdpt}
R.~Shaltiel.
\newblock Towards proving strong direct product theorems.
\newblock In {\em Proceedings of 16th IEEE Conference on Computational
  Complexity}, pp. 107--119, 2001.

\bibitem[Y83]{yao:prob} A.C.C.~Yao. 
\newblock Lower Bounds by Probabilistic Arguments. 
\newblock In {\em Proceedings of 24th IEEE Symposium on Foundations of Computer Science}, pp.~420--428, 1983.

\end{thebibliography}

\newcommand{\etalchar}[1]{$^{#1}$}

\end{document}